\documentclass[twoside]{article}
\usepackage{qic,epsfig}

\usepackage{amssymb}
\usepackage{amsmath}
\usepackage{amsthm}
\usepackage{graphicx}
\theoremstyle{plain}
\newtheorem{thm}{Theorem}
\newtheorem{problem}{Problem}
\newtheorem{defn}{Definition}
\newtheorem{rem}{Remark}
\newtheorem{lem}{Lemma}
\newtheorem{cor}{Corollary}
\newtheorem{fact}{Fact}

\begin{document}
\setlength{\textheight}{8.0truein}    
\global\long\global\long\global\long\def\bra#1{\mbox{\ensuremath{\langle#1|}}}
\global\long\global\long\global\long\def\ket#1{\mbox{\ensuremath{|#1\rangle}}}
\global\long\global\long\global\long\def\bk#1#2{\mbox{\ensuremath{\ensuremath{\langle#1|#2\rangle}}}}
\global\long\global\long\global\long\def\kb#1#2{\mbox{\ensuremath{\ensuremath{\ensuremath{|#1\rangle\!\langle#2|}}}}}

\runninghead{Universality of beamsplitters}
            {Adam Sawicki}

\normalsize\textlineskip
\thispagestyle{empty}
\setcounter{page}{1}

\copyrightheading{0}{0}{2003}{000--000}

\vspace*{0.88truein}

\alphfootnote

\fpage{1}

\centerline{\bf
Universality of beamsplitters}
\vspace*{0.37truein}
\centerline{\footnotesize
Adam Sawicki}
\vspace*{0.015truein}
\centerline{\footnotesize\it Center for Theoretical Physics, Massachusetts Institute of
Technology,}
\baselineskip=10pt
 \centerline{\footnotesize\it 77 Massachusetts Ave, Cambridge, MA 02139, USA}
 \vspace*{10pt}
\centerline{\footnotesize\it School of Mathematics, University of Bristol,}
\baselineskip=10pt
 \centerline{\footnotesize\it University Walk, Bristol BS8 1TW, UK}
\vspace*{10pt}
\centerline{\footnotesize\it Center for Theoretical Physics, Polish Academy of Sciences,}
\baselineskip=10pt
 \centerline{\footnotesize\it Al. Lotnik\'ow 32/46, 02-668 Warszawa, Poland}
\vspace*{10pt}
\publisher{(received date)}{(revised date)}

\vspace*{0.21truein}

\abstracts{We consider the problem of building an arbitrary $N\times N$ real orthogonal operator using a finite set, $S$, of elementary quantum optics gates operating on $m\leq N$ modes - the problem of universality of $S$ on $N$ modes. In particular, we focus on the universality problem of an $m$-mode beamsplitter. Using methods of control theory and some properties of rotations in three dimensions, we prove that any nontrivial real 2-mode and `almost' any nontrivial real $3$-mode beamsplitter is universal on $m\geq3$ modes.}{}{}

\vspace*{10pt}

\vspace*{3pt}
\communicate{to be filled by the Editorial}

\vspace*{1pt}\textlineskip    
\section{Introduction}        
Around twenty years ago Reck $et\, al.$ \cite{reck94} considered
the problem of building an arbitrary $N\times N$ unitary operator
using the set of elementary quantum optics gates. They gave a recursive
algorithm which transforms an $N\times N$ unitary matrix into an
arrangement of $2$-mode beamsplitters, phase shifters, and mirrors.
If one wants to implement any $N\times N$ unitary, however, then
one needs to have access to all possible beamsplitters and phase shifters.
This makes the results of \cite{reck94} not particularly useful in
a real experimental setting. Recently, Bouland and Aaronson \cite{BA14} considered
the same problem albeit with a finite set of optical gates available.
They showed that actually any single $2$-mode beamsplitter that nontrivially
acts on two modes densely generates unitary transformations $U(3)$
(or orthogonal transformations $O(3)$, in the real case).  Combining
their result with the arguments from \cite{reck94} they concluded
that any nontrivial $2$-mode beamsplitter densely generates $U(m)$
for $m\geq3$ modes. For example, a real $2$-mode beamsplitter is given by one rotation angle $\theta\in[0,2\pi]$. Having such a beamsplitter we let it operate on pairs of three available modes. This way we get operators which effectively mix $3$ modes. The resulting operators are dense in $O(3)$.  As a direct consequence, for building orthogonal (unitary)
transformation one does not need tunable beamsplitters - those whose $\theta$ can be changed. In fact any
$2$-mode beamsplitter with a fixed $\theta\notin\{0,\pi/2,\pi, 3\pi/2\}$ is universal for generation of quantum linear
optics. They left the problem of classifying optical gates that act
on three or more modes open. This kind of gates can be easily built
experimentally using coupled optical waveguides \cite{obrien08,silberger09}. We also note that there are some interesting developments in models of fermionic linear optics \cite{Oszman,Terhal}.
In this paper we address the open problem given in \cite{BA14} using methods that are
orthogonal to those used in \cite{BA14}. In a generic case, our approach can be used to obtain the desired classification for any number of modes which is impossible using representation theory arguments of \cite{BA14}. We also give another proof
of a $2$-mode real beamsplitter universality which is entirely based
on our method and does not make any use of the results of \cite{reck94,BA14}. 

Mathematically, a real $m$-mode beamsplitter is represented by an
orthogonal $m\times m$ matrix. Throughout this paper we will follow
convention of Nielsen and Chuang \cite{Chung00} and assume that this
matrix has a determinant equal to one (some authors use a convention
with $-1$ \cite{BA14}). Under this assumption the set of $m$-mode
beamsplitters forms the group $SO(m)$. 

We say that a finite set of beamsplitters $S\subset SO(m)$ is universal
on $m$ modes if and only if it generates $SO(m)$, i.e. any $m$-mode
beamsplitter can be approximated by a sequence of elements form $S$
(or their inverses) with an arbitrary precision. This definition is
analogous to the one for quantum gate universality. For example the
famous set consisting of $H$ and $T$ gates is universal for one
qubit as it generates any operation from $SU(2)$ with an arbitrary
precision \cite{Boykin99}. In this paper we consider the problem
of $m$-mode universality when the set $S$ is constructed (in some
natural way which we explain in the subsequent sections) from a single
$2$- or 3-mode beamsplitter and show that:
\begin{enumerate}
\item Any nontrivial 2-mode beamsplitter is universal on $m\geq3$ modes.
\item Almost any nontrivial $3$-mode beamsplitter is universal on $m\geq3$
modes.
\end{enumerate}
We also make several interesting statements for $m$-mode beamsplitters
for an arbitrary $m$. They concern universality on $k> m$ provided universality on $m$ modes. The method we use to obtain the result is a combination
of the fundamental theorem of control theory (see Theorem \ref{thm:main})
and some algebraic properties of the rotation group in three dimensions.
Recently, Bouland and Aaronson \cite{BA14} gave a proof of universality
of a $2$-mode beamsplitter (also for the case when $SO(2)$ is replaced
by $SU(2)$). Their approach is based on representation theory and
classification of subgroups of $SU(3)$. As the authors of \cite{BA14}
point out, such classification is missing starting from $SU(5)$ and
therefore their approach has clear limitations. The method presented
in this paper is complementary to \cite{BA14} and attacks the problem
from a different direction. It divides the problem into two. For $O_{m}\in SO(m)$
with the spectrum $\sigma(O_{m})=\left\{ e^{i\phi_{1}},e^{-i\phi_{1}},\ldots,e^{i\phi_{m/2}},e^{-i\phi_{m/2}}\right\} $,
where each $\phi_{i}$ is an irrational multiple of $\pi$ it boils
down to proving that some particular elements generate the special orthogonal Lie algebra. This can be fully handled. For $O_{m}$ with
at least one $\phi_{i}=a\pi$, where $a\in\mathbb{Q}$ the subtle
techniques to show that the product of two group elements that have
finite order can have infinite order are required. We discuss this
by considering an example with two rotations about the
$x$ and $z$ axes by rational angle $\theta$ whose $\cos(\theta)$ is algebraic of degree $2$ in Section \ref{sub:example}. The techniques used in this section are based on cyclotomic polynomials and they were used in the similar context in \cite{Boykin99}. They, however, do not generalise easily for an arbitrary rational angle. Therefore in the general case we use the recent results of Conway, Radin and Sadun concerning products of rotations \cite{CRS00,CRS99, RS99,RS98}.

Our method reveals the importance of mode permutations for $m$-mode beamsplitters,
$m\geq3$. The central role is played by the set $S(O_{m})=\left\{ P_{\sigma}^{t}O_{m}P_{\sigma}:\,\sigma\in S_{m}\right\} $$ $
where $P_{\sigma}$ are $m\times m$ matrices that permute modes of
the considered beamsplitter. In particular we show that in a great
number of cases universality of an $m$-mode beamsplitter $O_{m}$
on $k\geq m$ modes reduces to showing that the set $S(O_{m})$ is
universal on $m$-modes (rather than on $m+1$ modes). Moreover, we
show that already on $3$-modes there is a beamsplitter that is not
universal on $3$ and $4$ modes. It corresponds to what we call the
trivial action of permutation group, that is, to $S(O_{3})=\left\{ O_{3},\, O_{3}^{-1}\right\} $. 

The paper is organized as follows. In Section \ref{sec:method} we
discuss general aspects of the method we use in this paper. Next,
in Section \ref{sec:Universality-of2} we prove universality of a
nontrivial $2$-mode beamsplitter on $m\geq3$ modes. The proof is
divided into two parts. The first one is an elegant Lie-algebra calculation.
The second is showing that the product of two finite order orthogonal
rotations is a rotation by an angle which is an irrational multiple
of $\pi$. In the subsequent section we discuss some aspects of beamsplitters
operating on higher number of modes, introduce $S(O_{m})$ and prove
general results concerning $S(O_{m})$. Finally in Section \ref{sec:3-mode-beamsplitter}
the $3$-mode beamsplitters are discussed in details.

\section{The method\label{sec:method}}

In this section we sketch the method we will use throughout the paper.

Let $G$ be a connected Lie group and $\mathfrak{g}$ its Lie algebra.
We say that elements $\{g_{1},\ldots,g_{k}\}\subset G$ generate $G$,
if and only if the set
\begin{gather*}
<g_{1},\ldots,g_{k}>:=\{g_{a_{1}}^{k_{1}}\cdot g_{a_{2}}^{k_{2}}\cdots g_{a_{n}}^{k_{n}}:\, a_{i}\in\left\{ 1,\ldots,k\right\} ,\, k_{i}\in\mathbb{Z},\, n\in\mathbb{N}\},
\end{gather*}
is dense in $G$ that is $G=\overline{<g_{1,},\ldots,\, g_{k}>}$.
Similarly we say that subgroups $\left\{ H_{1},\ldots,\, H_{k}\right\} $
of $G$ generate $G$ iff the set
\begin{gather*}
<H_{1},\ldots,\, H_{k}>:=\{g_{a_{1}}^{k_{1}}\cdot g_{a_{2}}^{k_{2}}\cdots g_{a_{n}}^{k_{n}}:\, a_{i}\in\left\{ 1,\ldots,k\right\} ,\, g_{a_{i}}\in H_{a_{i}},\, k_{i}\in\mathbb{Z},\, n\in\mathbb{N}\},
\end{gather*}
is dense in $G$. Finally, let $S=\left\{ X_{1},\ldots,X_{k}\right\} \subset\mathfrak{g}$
be a subset of Lie algebra $\mathfrak{g}$. We say that $S$ generates
$\mathfrak{g}$ iff any element $X$ of $\mathfrak{g}$ can be expressed
as a linear combination of $X_{i}$'s and arbitrarily nested commutators
of $X_{i}$'s:

\begin{gather*}
X=\sum_{j}\alpha_{i}X_{i}+\sum_{i,j}\alpha_{ij}\left[X_{i},\, X_{j}\right]+\sum_{i,j,k}\alpha_{ijk}\left[X_{i},\left[X_{j},\, X_{k}\right]\right]+\ldots
\end{gather*}
The following theorem \cite{J97} will be of the great importance
in this paper.
\begin{thm}
\label{thm:main}Let $G$ be a connected Lie group and $\mathfrak{g}$
its Lie algebra. $G$ is generated by one-parameter subgroups $\{e^{tX}:\, t\in\mathbb{R}\}$,
$X\in S$ where $S$ is a finite subset of $\mathfrak{g}$ if and
only if $S$ generates $\mathfrak{g}$ as a Lie algebra.
\end{thm}
The problem which we are going to deal with is the following one:
\begin{problem}
\label{main_problem}Let $A=\{a_{1},\mbox{\ensuremath{\ldots}},a_{k}\}\subset G$
be a finite subset of $G$. We want to show that $A$ generates $G$,
that is, the group generated by $A$ is dense in $G$. 
\end{problem}
We make use of Theorem \ref{thm:main} to solve Problem \ref{main_problem}.
To this end we note that one can always write
\begin{gather*}
a_{i}=e^{X_{i}},\, X_{i}\in\mathfrak{g}.
\end{gather*}
We consider two cases: 
\begin{enumerate}
\item Assume that for all $i\in\{1,\ldots,k\}$ we have that $<a_{i}>$
is dense in $\{e^{tX_{i}}:\, t\in\mathbb{R}\}$, i.e. $\overline{<a_{i}>}=\{e^{tX_{i}}:\, t\in\mathbb{R}\}$.
Under this assumption, by Theorem \ref{thm:main} we get that $A$
generates $G$ if and only if $\{X_{1},\ldots,X_{k}\}$ generates
$\mathfrak{g}$.
\item Assume that for some $a_{i}$ the group $<a_{i}>$ is not dense in
$\{e^{tX_{i}}:\, t\in\mathbb{R}\}$ but $\{X_{1},\ldots,X_{k}\}$
generate $\mathfrak{g}$. In this case we cannot directly apply Theorem
\ref{thm:main}. What we can do however is to replace $a_{i}$ by
some element $b_{i}=e^{Y_{i}}$ that belong to $<A>$ and is such
that $<b_{i}>$ is dense in $\{e^{tY_{i}}:\, t\in\mathbb{R}\}$ and
$\{X_{1},\ldots,Y_{i},\ldots,X_{k}\}$ generate $\mathfrak{g}$. If
this kind of manipulation can be done for each ``bad'' $a_{i}\in A$
then the problem is solved by means of Theorem \ref{thm:main}.
\end{enumerate}
In the following we will show that this is the case for beamsplitters.

\section{\label{sec:Universality-of2}Universality of a real 2-mode beamsplitter }

A $2$-mode beamsplitter is given by a matrix $g_{\theta}\in SO(2)$
of the form 
\begin{gather*}
g_{\theta}=\left(\begin{array}{cc}
\cos(\theta) & \sin(\theta)\\
-\sin(\theta) & \cos(\theta)
\end{array}\right),
\end{gather*}
where $\theta\in[0,\,2\pi[$. Let us first look at the spectrum of
$g_{\theta}$. The characteristic equation reads:
\begin{gather*}
\lambda^{2}-2\lambda\cos(\theta)+1=0.
\end{gather*}
And therefore spectrum is given by $\{e^{i\theta},\, e^{-i\theta}\}$.
We want to show that three matrices:

\begin{gather}
O_{1,2}=\left(\begin{array}{ccc}
\cos(\theta) & \sin(\theta) & 0\\
-\sin(\theta) & \cos(\theta) & 0\\
0 & 0 & 1
\end{array}\right),\,\, O_{1,3}=\left(\begin{array}{ccc}
\cos(\theta) & 0 & \sin(\theta)\\
0 & 1 & 0\\
-\sin(\theta) & 0 & \cos(\theta)
\end{array}\right),\label{eq:o3}\\
O_{2,3}=\left(\begin{array}{ccc}
1 & 0 & 0\\
0 & \cos(\theta) & \sin(\theta)\\
0 & -\sin(\theta) & \cos(\theta)
\end{array}\right),\nonumber 
\end{gather}
for a given and fixed $\theta\in[0,2\pi[\setminus\{0,\frac{\pi}{2},\pi,\frac{3}{2}\pi\}$
generate $SO(3)$\footnote{The excluded angles correspond to either permutation of modes or $\pm$ identity operation.}. Following the reasoning explained in Section \ref{sec:method}
we note that $O_{i,j}=e^{X_{i,j}}$, where 
\begin{gather*}
X_{1,2}=\left(\begin{array}{ccc}
0 & \theta & 0\\
-\theta & 0 & 0\\
0 & 0 & 0
\end{array}\right),\,\, X_{1,3}=\left(\begin{array}{ccc}
0 & 0 & \theta\\
0 & 0 & 0\\
-\theta & 0 & 0
\end{array}\right),\,\, X_{2,3}=\left(\begin{array}{ccc}
0 & 0 & 0\\
0 & 0 & \theta\\
0 & -\theta & 0
\end{array}\right).
\end{gather*}
It is obvious that $\left\{ X_{1,2},\, X_{1,3},X_{2,3}\right\} $
generate Lie algebra $\mathfrak{so}(3)$ iff $\theta\neq0$, as matrices 

\begin{gather*}
E_{1,2}=\left(\begin{array}{ccc}
0 & 1 & 0\\
-1 & 0 & 0\\
0 & 0 & 0
\end{array}\right),\,\, E_{1,3}=\left(\begin{array}{ccc}
0 & 0 & 1\\
0 & 0 & 0\\
-1 & 0 & 0
\end{array}\right),\,\, E_{2,3}=\left(\begin{array}{ccc}
0 & 0 & 0\\
0 & 0 & 1\\
0 & -1 & 0
\end{array}\right),
\end{gather*}
form a standard basis of $\mathfrak{so}(3)$. We also note that for
$\theta=a\pi$ where $a$ is irrational we have that $<O_{ij}>$ is
dense in $\{e^{tX_{ij}}:\, t\in\mathbb{R}\}$. Therefore using Theorem
\ref{thm:main} we have that $\{O_{1,2},\, O_{1,3},\, O_{2,3}\}$
generate $SO(3)$ for $\theta=a\pi$ with irrational $a$. We still
need to examine the case when $a\in\mathbb{Q}$ that is the case when
for all three groups $\overline{<O_{i,j}>}\neq\{e^{tX_{i,j}}:\, t\in\mathbb{R}\}$.
To this end we choose three new matrices that belong to $<O_{1,2},\, O_{1,3},\, O_{2,3}>$,
i.e.: 
\begin{gather*}
O_{1,2}O_{1,3},\, O_{1,2}O_{2,3},\, O_{1,3}O_{2,3}.
\end{gather*}
Our goal is to show: 
\begin{description}
\item[Statement 1.] The Lie algebra elements corresponding to $\left\{ O_{1,2}O_{1,3},\, O_{1,2}O_{2,3},\, O_{1,3}O_{2,3}\right\} $
form a basis of $\mathfrak{so}(3)$.
\item[Statement 2.]  The elements $\left\{ O_{1,2}O_{1,3},\, O_{1,2}O_{2,3},\, O_{1,3}O_{2,3}\right\} $
are rotations by angles that are not rational multiples of $\pi$,
or equivalently there is no power $k\in\mathbb{N}$ for which they
become identity matrix.
\end{description}
The proofs of these two statements are given in the next three subsections.

\subsection{\label{sub:example}Finite and infinite order elements - examples.}

In this section we consider examples showing that the product of two orthogonal rotations by an angle $\theta$ which is a rational multiple of $\pi$ can be a rotation by an angle $\alpha$ which is not a rational multiple of $\pi$. In particular we show that this is the case for all $\theta$'s whose $\cos(\theta)$ is algebraic of degree two. Our approach makes a heavy use of cyclotomic polynomials as they are useful in showing that for a complex number $e^{i\alpha}$ its argument $\alpha$ in not a rational multiple of $\pi$. These techniques were also used in \cite{Boykin99} to show that gates $H$ and $T$ generate $SU(2)$, where the product of two specific rotation was considered. Here we provide much more general discussion which reveals the natural limitations of this method. The main purpose of this section is to provide some hands on examples and explicit calculations. For an arbitrary $\theta$ we will use a different approach (see section \ref{sec:statment2}). 
\subsubsection{Cyclotomic ploynomials}
The method we use is based on some properties
of cyclotomic polynomials and therefore we start with their discussion
(see \cite{Dummit91} for more details). 
\begin{defn}
The $n$-th cyclotomic polynomial for $n\in\mathbb{N}$ is given by
\begin{gather*}
\Phi_{n}(x)=\prod_{\begin{array}{c}
1\leq k\leq n\\
\mathrm{gcd}(k,n)=1
\end{array}}\left(x-e^{2i\pi\frac{k}{n}}\right),
\end{gather*}
where $\mathrm{gcd}(k,n)$ is the greatest common divisor of $n$
and $k$. 
\end{defn}
\noindent Cyclotomic polynomials have several useful properties. In the following
we will use three of them: (1) cyclotomic polynomials are monic and irreducible over $\mathbb{\mathbb{\mathbb{Q}}}$, (2) Coefficients of cyclotomic polynomials are integers, (3)
For any $q\in\mathbb{N}$ we have
\begin{gather}
x^{q}-1=\prod_{d|q}\Phi_{d}(x).\label{eq:xq-decomposition}
\end{gather}
The first property is actually a nontrivial result due to Gauss (it
can be however easily proved for prime $n$ using Eisenstein criterion
for irreducibility) \cite{Dummit91}. The second one is a direct result
of 
\begin{gather*}
x^{q}-1=\prod_{1\leq k\leq q}\left(x-e^{2i\pi\frac{k}{q}}\right),
\end{gather*}
and the definition of a cyclotomic polynomial. For $\alpha\in\mathbb{C}$ we say that the monic irreducible polynomial
$m_{\alpha}(x)\in\mathbb{\mathbb{Q}}[x]$ is the minimal polynomial
for $\alpha$ over $\mathbb{Q}$ iff $m_{\alpha}(\alpha)=0$. If $\alpha\in\mathbb{C}$ has a minimal polynomial over $\mathbb{Q}$ then we call it algebraic. Otherwise it is transcendental. The algebraic degree of $\alpha$ is a degree of its minimal polynomial. Now we can state the main theorem (cf. lemma 3.4 of \cite{ST02}):
\begin{thm}\label{thm-cyclo}
The following two are equivalent: (1) $a\in\mathbb{Q}$, (2) the minimal
polynomial for $\alpha=e^{i2\pi a}$ over $\mathbb{Q}$ exists and
is cyclotomic.\end{thm}
\begin{proof}
If $a=p/q$ then $(e^{i2\pi a})^{q}=1$ and therefore the minimal
polynomial $m_{\alpha}(x)$ exists as $ $$\alpha$ satisfies $x^{q}-1=0$.
By (\ref{eq:xq-decomposition}) $ $we know that $\alpha$ is a root
of some $\Phi_{d}(x)$ where $d|q$. But $\Phi_{d}(x)$ is irreducible
and monic hence it is the minimal polynomial for $\alpha$ over $\mathbb{\mathbb{Q}}$.
Conversely, assume $m_{\alpha}$ exists and is cyclotomic. Then we
have
\begin{gather}
0=m_{\alpha}(\alpha)=\Phi_{n}(\alpha)=\Pi_{d|n}\Phi_{d}(\alpha)=\alpha^{n}-1.\label{eq:relation}
\end{gather}
By (\ref{eq:relation}) we get $e^{2i\pi cn}=1$ and hence $c\in\mathbb{Q}$.
\end{proof}

\subsubsection{Products of rotations}\label{rot-prod}

Theorem \ref{thm-cyclo} can be used in particular for showing
that for a given complex number $e^{i\alpha}$ the angle $\alpha$ is not a rational multiple of $\pi$. In this case it is enough to prove that
either the minimal polynomial does not exist or it exists and is not cyclotomic. In the following we use it to study the
composition of two rotations about orthogonal axes by an angle $\theta$ which is a rational multiple of $\pi$.  Let $O_{1,2}$ and $O_{2,3}$ be as in (\ref{eq:o3}). We
have
\begin{gather*}
\mathrm{tr}O_{1,2}O_{2,3}=2\cos(\theta)+\cos^{2}(\theta).
\end{gather*}
From the other hand $O_{1,2}O_{2,3}$ is a rotation by an angle $\alpha$
and hence $\mathrm{tr}O_{1,2}O_{2,3}=1+2\cos(\alpha)$. As a result
we get the equation which relates $\theta$ and $\alpha$:
\begin{gather}
2\cos(\alpha)=2\cos(\theta)+\cos^{2}(\theta)-1.\label{eq:phiteta}
\end{gather}
We first determine if $e^{i\alpha}$ is algebraic or transcendental. 
\begin{fact}
For $\alpha$ given by (\ref{eq:phiteta}) $e^{i\alpha}$ is an algebraic number.
\end{fact}\label{phi-algebraic}
\begin{proof}
As $\theta$ is a rational multiple of $\pi$, using De Moivre's formula we get that $\cos (\theta)$ is an algebraic number. Next, it is known \cite{ST02}, that the sum and product of two algebraic numbers is again algebraic number. Applying this to (\ref{eq:phiteta}) we get that $\cos(\alpha)$ is an algebraic number. A complex number is algebraic iff its both real and imaginary parts are algebraic. To show that $\sin(\alpha)$ is algebraic note that the field extensions $\mathbb{Q}\left[\cos(\alpha)\right]:\mathbb{Q}$ and $\mathbb{Q}\left[\cos(\alpha),\sin(\alpha)\right]:\mathbb{Q}\left[\cos(\alpha)\right]$  are both algebraic and consequently by the chain rule $\mathbb{Q}\left[\sin(\alpha)\right]:\mathbb{Q}$ is algebraic. Hence $e^{i\alpha}=\cos\alpha+i\sin\alpha$ is algebraic.
\end{proof}
\noindent Note that by Fact \ref{phi-algebraic} $e^{i\alpha}$ is never transcendental so in order to use Theorem \ref{thm-cyclo} we need to determine its minimal polynomial. Putting $x=e^{i\alpha}=\cos(\alpha)+i\sqrt{1-\cos^2(\alpha)}$ we get 
\begin{gather}\label{eq1}
x^{2}-2\cos(\alpha)x+1=0,
\end{gather}
where $\cos(\alpha)$ is determined by $\cos(\theta)$. 
\begin{fact}\label{degree}
The algebraic degree of $\cos(\alpha)$ divides the algebraic degree of $\cos(\theta)$ and the algebraic degree of $e^{i\alpha}$ is twice the algebraic degree of $\cos(\alpha)$.
\end{fact}
\begin{proof}
It is easy to see that the field extensions satisfy $\mathbb{Q}[\cos(\alpha)]\subset\mathbb{Q}[\cos(\theta)]$. Using the chain rule for fields \cite{ST02} we get
\begin{gather}
[\mathbb{Q}[\cos(\theta)]:\mathbb{Q}]=[\mathbb{Q}[\cos(\theta)]:\mathbb{Q}[\cos(\alpha)][\mathbb{Q}[\cos(\alpha)]:\mathbb{Q}]
\end{gather}
As $[\mathbb{Q}[\cos(\theta)]:\mathbb{Q}]$ and $[\mathbb{Q}[\cos(\alpha)]:\mathbb{Q}]$ are algebraic degrees of $\cos(\theta)$ and $\cos(\alpha)$ respectively we get the conclusion. The relation between algebraic degrees $e^{i\alpha}$ and $\cos(\alpha)$ comes from equation (\ref{eq1}). 
\end{proof}
Note that $\cos(\theta)$ can have arbitrary large algebraic degree and therefore minimal polynomial for $e^{i\alpha}$ can have any order. In the following we consider example when degree of $\cos(\theta)$ is $2$ which by Fact \ref{degree} means the algebraic degree of $e^{i\alpha}$ can be either $2$ or $4$.  As by our assumption $\cos(\theta)=a+b\sqrt{C}$ we get $cos(\alpha)=A+B\sqrt{C}$, where $A$ and $B$ are given in terms of $a,b,c$. The minimal polynomial of $e^{i\alpha }$ is:
\begin{gather}\label{min-pol}
x^4-4Ax^3+(4A^2+2)x^2-4Ax-4B^2C+1=0
\end{gather}
We next determine possible values of $a$ and $b$. Using De Moivre's formula one can easily see that roots $x_1$ and $x_2$ with $|x_1|<2$ and $|x_2| <2$ of 
\begin{gather}
x^2+ax+b=0,
\end{gather} 
where $a,b\in \mathbb{Z}$ are only possible values of $2\cos(\theta)$ of algebraic degree two. Direct calculation leads to: (1) $\cos(\theta)=\pm\frac{1}{\sqrt{2}}$, (2) $\cos(\theta)=\pm\frac{\sqrt{3}}{2}$, (3) $\cos(\theta)=\pm\frac{1}{4}\pm\frac{\sqrt{5}}{4}$. The corresponding angles are: $\cos(\pi/4)=1/\sqrt{2}$, $\cos(\pi/3)=\sqrt{3}/2$, $\cos(\pi/5)=1/4+1/4\sqrt{5}$, $\cos(2\pi/5)=-1/4+1/4\sqrt{5}$, $\cos(3\pi/5)=1/4-1/4\sqrt{5}$, $\cos(4\pi/5)=-1/4-1/4\sqrt{5}$, $\cos(3\pi/4)=-1/\sqrt{2}$, $\cos(2\pi/3)=-\sqrt{3}/2$. One explicitly checks that in all case  polynomial (\ref{min-pol}) is not cyclotomic. This way we showed that the
product of two rotations about orthogonal axes by  rational angle $\theta$, whose $\cos{\theta}$ is algebraic of degree two, is a rotation by an angle which is an irrational multiple of $\pi$. 

\subsection{The proof of Statement 1}\label{sub:linear3}

We will make use of a compact form of Baker-Campbell-Hausdorff (BCH)
formula for the group $SO(3)$. For the detailed derivation see \cite{BCH01} (cf. \cite{H853}).
Let us first recall the definition of the BCH formula for an arbitrary
compact semisimple matrix Lie algebra $\mathfrak{g}$. For $X,Y\in\mathfrak{g}$ we define
$\mbox{BCH}(X,Y)$ in the following way:

\begin{gather*}
e^{\mbox{BCH}(X,Y)}=e^{X}e^{Y}.
\end{gather*}
It is known that $\mbox{BCH}(X,Y)$ is given by an infinite sum. In
the case of $\mathfrak{so}(3)$, however, there is a particularly
nice compact formula for $\mbox{BCH}(X,Y)$. This is due to the following
two facts:
\begin{enumerate}
\item For $X\in\mathfrak{so}(3)$ the characteristic polynomial is given
by $p(\lambda)=-\lambda^{3}-||X||^{2}\lambda$, where $||X||^{2}=\frac{1}{2}\mathrm{tr}(X^{t}X)$.
Therefore by the Cayley-Hamilton theorem $X^{3}=-\theta^{2}X$, where
$\theta=||X||$. 
\begin{gather}
e^{X}=I+X+\frac{1}{2!}X^{2}+\frac{1}{3!}X^{3}+\ldots+\frac{1}{n!}X^{n}+\ldots=\nonumber \\
=I+X\left(1-\frac{\theta^{2}}{3!}+\frac{\theta^{4}}{5!}-\frac{\theta^{6}}{7!}+\ldots\right)+X^{2}\left(\frac{1}{2!}-\frac{\theta^{2}}{4!}+\frac{\theta^{4}}{6!}-\ldots\right)=\label{eq:rodriguez}\\
=I+\frac{\sin\theta}{\theta}X+2\frac{\sin^{2}(\theta/2)}{\theta^{2}}X^{2}.\nonumber 
\end{gather}

\item Let $O=e^{X}$ be a rotation matrix from $SO(3)$ and let $Z=\frac{O-O^{t}}{2}$.
Using formula (\ref{eq:rodriguez}) we see that $Z=\frac{\sin\theta}{\theta}X$.
Therefore: 
\begin{gather}
\log(O):=X=\frac{\sin^{-1}(||Z||)}{||Z||}Z.\label{eq:log}
\end{gather}

\end{enumerate}
Having these the BCH formula for $\mathfrak{so}(3)$ can be easily
calculated. One simply writes down the expression for $e^{X}e^{Y}$
using (\ref{eq:rodriguez}) and then calculates the logarithm using
(\ref{eq:log}). Details can be found in \cite{BCH01}. In the case
where $X$ and $Y$ are orthogonal $\mbox{tr}(X^{t}Y)=0$ the formula
reads:
\begin{gather}
\mbox{BCH}(X,Y)=\alpha X+\beta Y+\gamma[X,Y],\label{eq:bch}
\end{gather}
where
\begin{gather}
\alpha=\frac{\sin^{-1}(d)}{d}\frac{a}{\theta},\,\,\beta=\frac{\sin^{-1}(d)}{d}\frac{b}{\phi},\,\,\gamma=\frac{\sin^{-1}(d)}{d}\frac{c}{\theta\phi},\label{eq:BCH-1}\\
a=\sin\theta\cos^{2}\left(\phi/2\right),\,\, b=\mbox{sin}\phi\cdot\mbox{cos}^{2}\left(\theta/2\right),\,\, c=\frac{1}{2}\sin\theta\sin\phi,\nonumber \\
d=\sqrt{a^{2}+b^{2}+c^{2}},\,\,\theta=||X||,\,\,\phi=||Y||.\nonumber 
\end{gather}
We need to show that
\begin{gather}
\left\{ \mbox{BCH}(X_{1,2},X_{1,3}),\,\mbox{BCH}(X_{1,2},X_{2,3}),\,\mbox{BCH}(X_{1,3},X_{2,3})\right\} ,\label{bch11}
\end{gather}
form a basis of $\mathfrak{so}(3)$. Using (\ref{eq:bch}) and (\ref{eq:BCH-1})
we easily find:
\begin{gather}
\mbox{BCH}(X_{1,2},X_{1,3})=\frac{\sin^{-1}(d)}{d\cdot\theta}\sin\theta\left(\cos^{2}\frac{\theta}{2}X_{1,2}+\cos^{2}\frac{\theta}{2}X_{1,3}-\frac{1}{2}\sin\theta X_{2,3}\right),\label{eq:bchel}\\
\mbox{BCH}(X_{1,2},X_{2,3})=\frac{\sin^{-1}(d)}{d\cdot\theta}\sin\theta\left(\cos^{2}\frac{\theta}{2}X_{1,2}+\cos^{2}\frac{\theta}{2}X_{2,3}+\frac{1}{2}\sin\theta X_{1,3}\right),\nonumber \\
\mbox{BCH}(X_{1,3},X_{2,3})=\frac{\sin^{-1}(d)}{d\cdot\theta}\sin\theta\left(\cos^{2}\frac{\theta}{2}X_{1,3}+\cos^{2}\frac{\theta}{2}X_{2,3}-\frac{1}{2}\sin\theta X_{1,2}\right).\nonumber 
\end{gather}
where $d=\sin\theta\sqrt{2\cos^{4}(\theta/2)+1/4\sin^{2}\theta}$.
Next we write down matrices (\ref{eq:bchel}) in the standard basis
of $\mathfrak{so}(3)$ and get the following change of basis matrix

\begin{gather}
M=\frac{\sin^{-1}(d)}{d\cdot\theta}\sin\theta\cdot A_{SO(3)}=\frac{\sin^{-1}(d)}{d}\sin\theta\left(\begin{array}{ccc}
\cos^{2}\frac{\theta}{2} & \cos^{2}\frac{\theta}{2} & -\frac{1}{2}\sin\theta\\
\cos^{2}\frac{\theta}{2} & \frac{1}{2}\sin\theta & \cos^{2}\frac{\theta}{2}\\
-\frac{1}{2}\sin\theta & \cos^{2}\frac{\theta}{2} & \cos^{2}\frac{\theta}{2}
\end{array}\right).\label{det}
\end{gather}
The determinant of (\ref{det}) is given by:

\begin{gather}\label{e1}
\left(\frac{\sin^{-1}(d)}{d}\sin\theta\right)^{3}\left(-2\cos^{6}(\theta/2)+\frac{1}{2}\cos^{4}(\theta/2)\sin(\theta)+\frac{1}{8}\sin^{3}(\theta)\right)
\end{gather}
To find its zeros of (\ref{e1}) one can for example write down all functions in terms of $t=\mathrm{tan}(\frac{\theta}{4})$ and solve the polynomial equations with respect to $t$.The relevant part of (\ref{e1}) reads:
\begin{gather}
-2\cos^{6}(\theta/2)+\frac{1}{2}\cos^{4}(\theta/2)\sin(\theta)+\frac{1}{8}\sin^{3}(\theta)=0.
\end{gather}
using $\sin\theta=2\sin\frac{\theta}{2}\cos\frac{\theta}{2}$ one
gets: $ $
\begin{gather}\label{e2}
\cos^{3}(\theta/2)\left(-2\cos^{3}(\theta/2)+\cos^{2}(\theta/2)\sin(\theta/2)+\sin^{3}(\theta/2)\right)=0.
\end{gather}
The first factor of (\ref{e2}) gives $\theta=\pi$. Using $\cos\theta/2=\frac{1-t^{2}}{1+t^{2}}$
and $\sin\theta/2=\frac{2t}{1+t^{2}}$, for the second factor of (\ref{e2}) one
gets polynomial equation that has only one positive real root $t=\sqrt{2}-1$
(the remaining roots are $t=-\sqrt{2}-1$ and four complex roots). That means $\theta=\pi/2$. Therefore the determinant (\ref{e1}) vanishes iff $\theta=0$, $\theta=\pi$ or
$\theta=\pi/2$. Hence (\ref{bch11}) form a basis of $\mathfrak{so}(3)$
in all cases we are interested in.
\begin{rem}
Using isomorphism $SO(3)=SU(2)/\mathbb{Z}_{2}$ one can also prove
this result working with $SU(2)$ matrices. 
\end{rem}

\subsection{The proof of Statement 2}\label{sec:statment2}

For the proof we use the following result of \cite{RS99}:
\begin{lem}
\label{lem:rotations}Let $A$ and $B$ be rotations about orthogonal
axes by $2\pi/p$ and $2\pi/q$ respectively. Consider the word:
\begin{gather*}
A^{a_{1}}B^{b_{1}}\cdots A^{a_{n}}B^{b_{n}},
\end{gather*}
where none of $a_{i}$'s are multiple of $p/2$ and none of $b_{i}$'s
are multiple of $q/2$. If no two consecutive terms represent rotations
by $\pi/2$ and if the word is nonempty, then the word is not equal
to identity. \end{lem}
\begin{cor}
\label{thm:rational}Let $O_{1},O_{2}\in SO(3)$ be rotations about
orthogonal axes by $2\pi\frac{a}{p}$ and $2\pi\frac{b}{q}$, respectively,
where $a<p$ and $b<q$ and fractions $a/p$ and $b/q$ are not equal
$\{1,\,\frac{1}{2},\,\frac{1}{4},\,\frac{3}{4}\}$. Then for any $n\in\mathbb{N}$
we have $\left(O_{1}O_{2}\right)^{n}\neq I$. Thus the product of
two finite order orthogonal rotations is a rotation by an angle which
is not a rational multiple of $\pi$. \end{cor}
\begin{proof}
Let $A$ and $B$ be orthogonal rotations by $\frac{2\pi}{p}$ and
$\frac{2\pi}{q}$ respectively. We put $O_{1}=A^{a}$ and $O_{2}=B^{b}$.
Then $\left(O_{1}O_{2}\right)^{n}=\left(A^{a}B^{b}\right)^{n}.$ By
the assumptions $a\neq\frac{p}{2}$ and $b\neq\frac{q}{2}$. Moreover,
neither $O_{1}$ nor $O_{2}$ is a rotation by $\pi/2$. Therefore,
using Lemma \ref{lem:rotations} we get $\left(A^{a}B^{b}\right)^{n}\neq1$. 
\end{proof}
Combining corollary \ref{thm:rational} with the results of Section
\ref{sub:linear3} we get:
\begin{thm}
Assume $\theta\notin\{0,\,\frac{\pi}{2},\,\pi,\,\frac{3}{2}\pi\}$.
Let
\begin{gather*}
S=\left\{ O_{k,l}:\, O_{k,l}\in SO(3),\, k,l\in\left\{ 1,\,2,\,3\right\} ,\, k<l\right\} ,
\end{gather*}
where $O_{k,l}$ are given by (\ref{eq:o3}). The set S generates
$SO(3)$ and hence a real 2-mode beamsplitter is universal on 3-modes. 
\end{thm}
$ $
\subsection{A real 2-mode beamsplitter is universal on $m$ modes, $m\geq3$}

In the previous section we discussed the first non-trivial case, that
is, $SO(3)$. It turns out that having the result for $SO(3)$ is
almost enough to state the corresponding one for $SO(N)$.  To this
end let us denote by $\{\ket i\}_{i=1}^{N}$ a basis of $\mathbb{C}^{N}$.
In this section we will consider the following set
\begin{gather}\label{set-S}
S=\left\{ O_{k,l}:\, O_{k,l}\in SO(N),\, k,l\in\left\{ 1,\ldots,N\right\} ,\, k<l\right\} ,
\end{gather}
 where $O_{k,l}$ represent the matrix of a beamsplitter acting on
modes $k$, $l$:
\begin{gather*}
O_{kl}=\cos\theta\left(\kb kk+\kb ll\right)+\sin\theta\left(\kb kl-\kb lk\right).
\end{gather*}
The number of these matrices is $N(N-1)/2$. The Lie algebra elements
satisfying $e^{X_{kl}}=O_{kl}$ are given by
\begin{gather}\label{Xkl}
X_{kl}=\theta\left(\kb kl-\kb lk\right)=\theta E_{kl},
\end{gather}
where $\{E_{kl}\}_{k<l}$ is a standard basis of $\mathfrak{so}(N)$
and we have the following commutation relations:
\begin{gather*}
\left[X_{k,l},X_{k,m}\right]=-\theta X_{l,m},\,\,\,\left[X_{k,l},X_{l,m}\right]=\theta X_{k,m},\left[X_{k,m},X_{l,m}\right]=-\theta X_{k,m},
\end{gather*}
where we assumed $k<l$, $l<m$, $k<m$. 

We first, in section \ref{gen-arg} give an argument based purely on Theorem \ref{thm:main} which shows that generation of $SO(3)$ by $\left\{ O_{1,2},\,O_{2,3},\,O_{1,3}\right\}$ implies generation of $SO(N)$ by $S$ (defined as in (\ref{set-S})). This reasoning is then extended in section \ref{sec-mmode} to state some general facts about universality of $m$-mode beamsplitters, where $m>2$. In section \ref{spec-ex} we give another argument which is tailored specifically for $2$-mode beamsplitters. It has an interesting advantage over the general argument which we discuss in section \ref{gen-arg}.  
\subsubsection{General argument}\label{gen-arg}
\begin{fact}\label{general-2}Assume that $\left\{ O_{1,2},\,O_{2,3},\,O_{1,3}\right\}$ generates $SO(3)$. Then  \begin{gather}\left\{ O_{k,l}:\, O_{k,l}\in SO(N),\, k,l\in\left\{ 1,\ldots,N\right\} ,\, k<l\right\}\end{gather} generates $SO(N)$. 
\end{fact}
\begin{proof}
Assume that $\left\{ O_{1,2},\,O_{2,3},\,O_{1,3}\right\}$ generates $SO(3)$. In particular, it means that after closure one can obtain any $2$-mode beamsplitter $O\in SO(2)\subset SO(3)$ acting on any pair of available $3$ modes. Therefore we have (at least in the limit) access to elements $e^{E_{12}},\,e^{E_{13}},\,e^{E_{23}}$. Repeating that argument for any three out of $N$ modes we obtain all possible elements $e^{E_{kl}}$ with $1\leq k<l\leq N$- some are obtained more than once.  Note that set $\{E_{kl}\}_{k<l}$ is a standard basis of $\mathfrak{so}(N)$ and the rotation angle of $e^{E_{kl}}$ is $1$ which is clearly not a rational multiple of $\pi$. Therefore, by Theorem \ref{thm:main}, we get the desired result.
\end{proof}
In the proof of Fact \ref{general-2} we used that certain elements are available after closure. It is not clear, and in general may not be true, that they are available before closing set $<S>$. In the situation when the dense set generated by $\left\{ O_{1,2},\,O_{2,3},\,O_{1,3}\right\}$ does not contain elements $e^{E_{12}},\,e^{E_{13}},\,e^{E_{23}}$ the argument described  above uses  elements that are available only in the approximate sense to show generation of $SO(N)$. From the mathematical point of view this is not a problem. One can also say that perhaps other elements $e^{\phi_{12}E_{12}},\,e^{\phi_{13}E_{13}},\,e^{\phi_{23}E_{23}}$ with $\phi_{ij}$ not rational multiples of $\pi$ are in fact available in $< O_{1,2},\,O_{2,3},\,O_{1,3}>$. This can be true, however, it is not clear which $\phi_{ij}$'s are possible and which are not. Therefore, the proof does not give any insight into how elements of $SO(N)$ are generated. Form the practical point of view one would like to know at least one example of a small number of elements that belong to $<S>$, that can be constructed in a simple way from the available beamsplitter and that enable generation of any element of $SO(N)$. In section \ref{spec-ex} we show an exemplary construction which provides the set of elements  $S^{(N)}$ that generates a dense set in $SO(N)$ and satisfies the assumptions of theorem \ref{thm:main}. Moreover these elements are available in $<S>$ before any closure and they are given by products of the original beamsplitter acting on selected triplets of modes. The construction can be viewed as an alternative proof of Fact \ref{general-2}.
\subsubsection{The example}\label{spec-ex}

If $\theta\neq0$ then $\left\{ X_{kl}\right\} _{k<l}$,where $X_{kl}$ is given by (\ref{Xkl}), spans Lie algebra
$\mathfrak{so}(3)$. Moreover when $\theta=a\pi$, $a\notin\mathbb{Q}$
then $<O_{k,l}>$ is dense in $\left\{ e^{tX_{k,l}}:\, t\in R\right\} $.
Therefore we can use Theorem \ref{thm:main} and obtain that $\left\{ O_{kl}\right\} _{k<l}$
generates dense subset of $SO(N)$. What is left is to consider the
case when $a\in\mathbb{Q}$. Note that for $k<l$ and $m<n$, we have
$O_{k,l}O_{m,n}=O_{m,n}O_{k,l}$ iff $\{k,l\}\cap\{m,n\}=\emptyset$.
In this case non-trivial elements of the spectrum (i.e. those different
from 1) of $O_{k,l}O_{m,n}$ are nontrivial elements of spectra of
$O_{k,l}$ and $O_{m,n}$ and hence if $a$ is rational they are rational
multiples of $\pi$. Thus we are interested only in the case when
$\{k,l\}\cap\{m,n\}\neq0$. Without any loss of generality we can
assume that $k\leq m$. We have three possibilities:
\begin{gather}
k=m:\,\, O_{k,l}O_{k,n},\nonumber \\
l=m:\,\, O_{k,l}O_{l,n},\label{eq:mnoz}\\
l=n:\,\, O_{k,l}O_{m,l}.\nonumber 
\end{gather}
Consider now the isomorphism: $\ket k\mapsto\ket 1$, $\ket l\mapsto\ket 2$,
$\ket m\mapsto\ket 3$( $\ket n\mapsto\ket 3$) between the $3$-dimensional
spaces: $\mbox{Span}_{\mathbb{C}}\{\ket k,\ket l,\ket m(\ket n)\}$
and $\mathbb{C}^{3}$. Under this isomorphism we can apply Theorem
\ref{thm:rational} and obtain that spectra of matrices  (\ref{eq:mnoz})
are irrational multiples of $\pi$. Moreover, we can use formulas
(\ref{eq:BCH-1}) to find the corresponding BCH elements belonging
to $\mathfrak{so}(N)$:
\begin{gather*}
\mbox{BCH}(X_{k,l},X_{k,n})=\frac{\sin^{-1}(d)}{d\cdot\theta}\sin\theta\left(\cos^{2}\frac{\theta}{2}X_{k,l}+\cos^{2}\frac{\theta}{2}X_{k,n}-\frac{1}{2}\sin\theta X_{l,n}\right),\\
\mbox{BCH}(X_{k,l},X_{l,n})=\frac{\sin^{-1}(d)}{d\cdot\theta}\sin\theta\left(\cos^{2}\frac{\theta}{2}X_{k,l}+\cos^{2}\frac{\theta}{2}X_{l,n}+\frac{1}{2}\sin\theta X_{k,n}\right),\\
\mbox{BCH}(X_{k,m},X_{l,m})=\frac{\sin^{-1}(d)}{d\cdot\theta}\sin\theta\left(\cos^{2}\frac{\theta}{2}X_{k,m}+\cos^{2}\frac{\theta}{2}X_{k,m}-\frac{1}{2}\sin\theta X_{k,m}\right),
\end{gather*}
where $d=\sin\theta\sqrt{2\cos^{4}(\theta/2)+1/4\sin^{2}\theta}$.
Next we explain how to chose indices $\{k,l\}$ and $\{m,n\}$ so
that the corresponding BCH elements generate $\mathfrak{so}(N)$.
This is in fact the nontrivial part of the extension from $\mathfrak{so}(3)$
to $\mathfrak{so}(N)$, $N>3$.

\begin{figure}[h]
\begin{center}
\epsfig{file=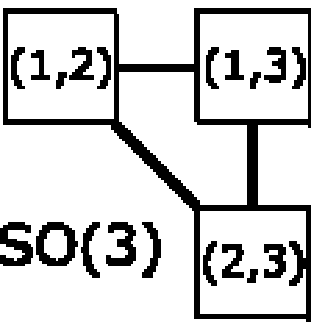, scale=0.55}\quad{}\quad{}\quad{}\quad{}\quad{}\epsfig{file=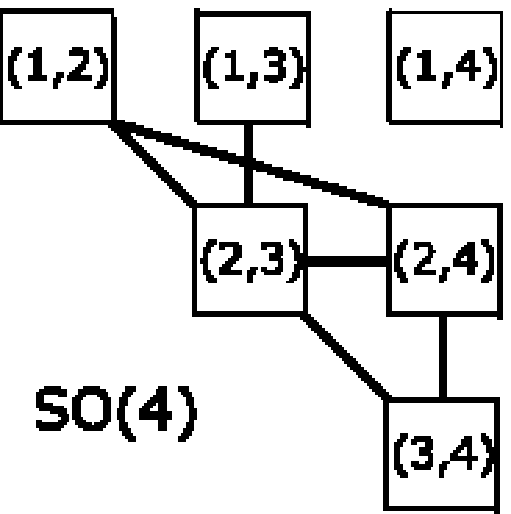, scale=0.55}\quad{}\quad{}\quad{}\quad{}\quad{}\epsfig{file=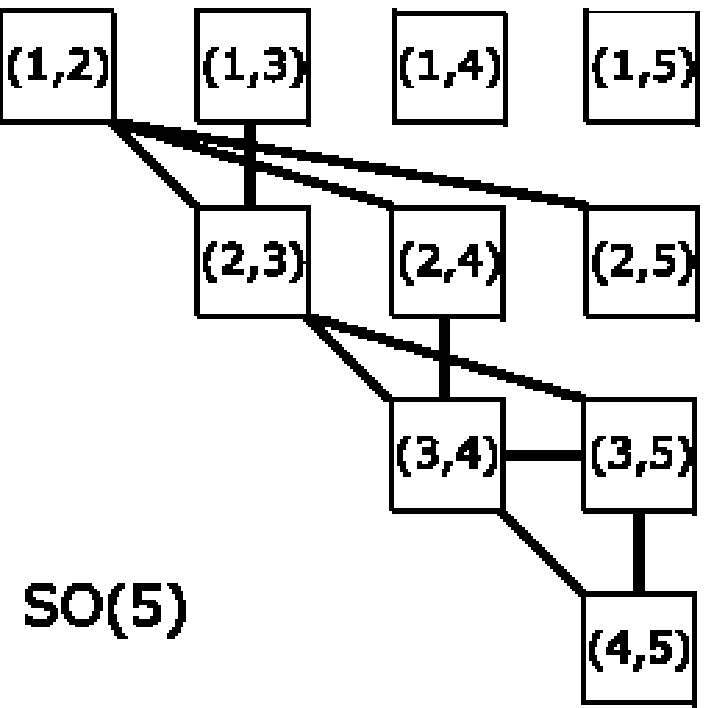, scale=0.55}
\bigskip{}
\bigskip{}
\bigskip{}
\bigskip{}
\qquad{}\qquad{}\qquad{}\qquad{}\epsfig{file=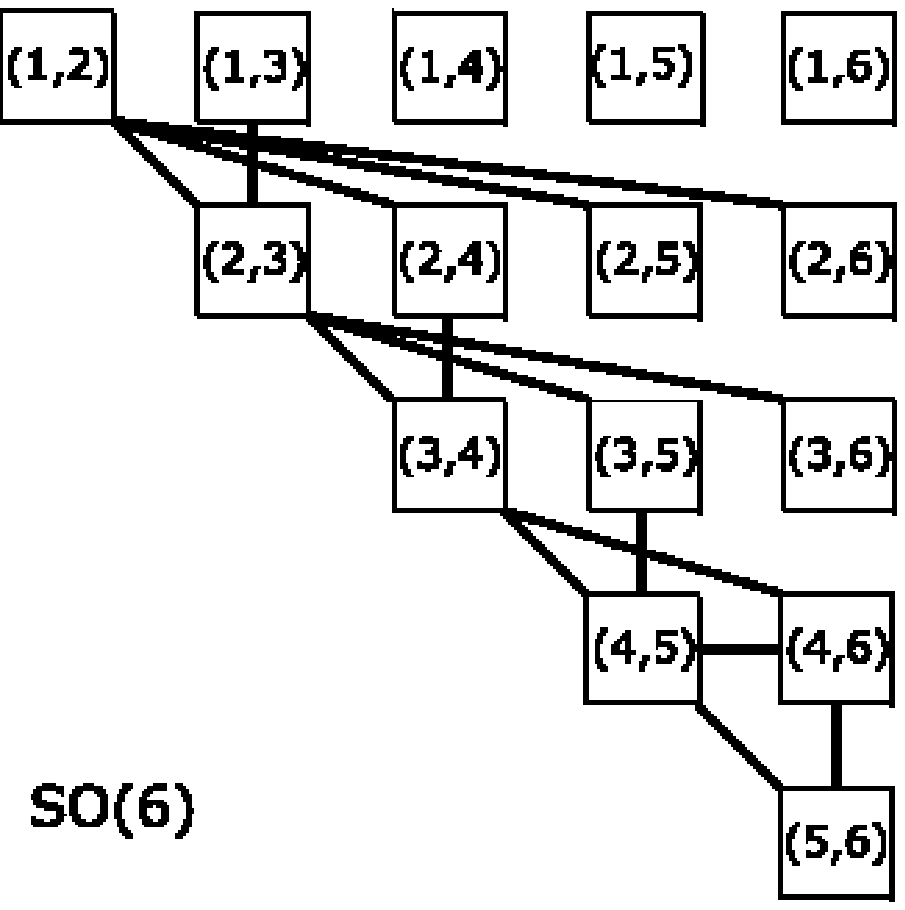, scale=0.55}
\end{center}
\fcaption{The pictorial representation of the rules for choosing elements $O_{k,l}O_{m,n}$. An index $(k,l)$ represent the matrix $O_{k,l}$. If indices $(k,l)$ and $(m,n)$ are connected by a line then matrix $O_{k,l}O_{m,n}$ is chosen as an element of a new generating set for $SO(N)$.\label{fig:The-pictorial-representation}}
\end{figure}


In case of $SO(3)$ we chosen as a new generating set
\begin{gather*}
S^{(3)}=\left\{ O_{12}O_{13},\, O_{12}O_{23},\, O_{13}O_{23}\right\} ,
\end{gather*}
and showed that corresponding $\mbox{BCH}$ elements form the basis
of $\mathfrak{so}(3)$. The construction of the basis for $N>3$ is
presented pictorially in figure \ref{fig:The-pictorial-representation}.
Each ``box'' with an index $(i,j)$ represents the matrix $O_{i,j}$.
If there is a line between two boxes $(k,l)$ and $(m,n)$, where
$k\leq m$ the product $O_{k,l}O_{m,n}$ is a member of a new generating
set for $SO(N)$. For example in case of $SO(4)$ we have:

\begin{gather*}
S^{(4)}=\left\{ O_{2,3}O_{2,4},\, O_{2,3}O_{3,4},\, O_{2,4}O_{3,4}\right\} \cup\left\{ O_{1,2}O_{2,3},\, O_{1,2}O_{2,4},O_{1,3}O_{2,3}\right\} \simeq\\
\simeq S^{(3)}\cup R^{(4)},
\end{gather*}
where $R^{(4)}=\left\{ O_{1,2}O_{2,3},\, O_{1,2}O_{2,4},O_{1,3}O_{2,3}\right\} $.
The first set, $S^{(3)}$, is isomorphic to the one we used for $SO(3)$
(under isomorphism $\ket 1\mapsto\ket 2$, $\ket 2\mapsto\ket 3$,
$\ket 3\mapsto\ket 4$). We need to show that $\mbox{BCH}$ elements
\begin{gather*}
\left\{ \mbox{BCH}(X_{2,3}X_{2,4}),\,\mbox{BCH}(X_{2,3}X_{3,4}),\,\mbox{BCH}(X_{2,4}X_{3,4})\right\} \cup\\
\cup\left\{ \mbox{BCH}(X_{1,2}X_{2,3}),\,\mbox{BCH}(O_{1,2}O_{2,4}),\,\mbox{BCH}(O_{1,3}O_{2,3})\right\} ,
\end{gather*}
are linearly independent. To this end we write them in the standard
basis $\left\{ E_{k,l}\right\} _{k<l}$ of $\mathfrak{so}(4)$. The
corresponding coefficients form columns of the change of basis matrix.
We use the ordering of the basis elements that reflects the structure
of $S^{(4)}$ that is $\left\{ E_{2,3},\, E_{2,4},\, E_{3,4}\right\} \cup\left\{ E_{12},\, E_{13},\, E_{14}\right\} $.
Under these assumption the change of basis matrix has the following
structure:
\begin{gather*}
\frac{\sin^{-1}(d)}{d\cdot\theta}\sin\theta\cdot A_{SO(4)}=\frac{\sin^{-1}(d)}{d\cdot\theta}\sin\theta\left(\begin{array}{cc}
A_{SO(3)} & N\\
0_{3\times3} & P_{SO(4)}
\end{array}\right),
\end{gather*}
where
\begin{gather*}
N=\left(\begin{array}{ccc}
\cos^{2}\frac{\theta}{2} & 0 & \cos^{2}\frac{\theta}{2}\\
0 & \cos^{2}\frac{\theta}{2} & 0\\
0 & 0 & 0
\end{array}\right),\,\, P_{SO(4)}=\left(\begin{array}{ccc}
\cos^{2}\frac{\theta}{2} & \cos^{2}\frac{\theta}{2} & -\frac{1}{2}\sin\theta\\
\frac{1}{2}\sin\theta & 0 & \cos^{2}\frac{\theta}{2}\\
0 & \frac{1}{2}\sin\theta & 0
\end{array}\right).
\end{gather*}
Since we are interested in the determinant only, the matrix $N$ is
irrelevant for calculation of $\mbox{det}(A_{SO(4)})=\mbox{det}(A_{SO(3)})\cdot\mathrm{det}(P_{SO(4)})$.
But $\det(P_{SO(4)})=-\frac{1}{2}\sin\theta\left(\frac{1}{4}\sin^{2}\theta+\cos^{4}\frac{\theta}{2}\right)$
and hence vanishes iff $\theta=0$ or $\theta=\pi$.

\paragraph{General $N$}

To prove that BCH elements corresponding to products of rotations
chosen according to the rules explained in figure 1 form the basis
of $\mathfrak{so}(N)$ we proceed by induction. Assume that $\mbox{det}(A_{SO(N-1)})$
is nontrivial except for $\theta\in\{0,\,\frac{\pi}{2},\,\pi,\,\frac{3}{2}\pi\}$.
Next note that $S^{(N)}\simeq S^{(N-1)}\cup R^{(N)}$ where
\begin{gather*}
R^{(N)}=\left\{ O_{1,2}O_{2,3},\, O_{1,2}O_{2,4},\ldots,\, O_{1,2}O_{2,N},\, O_{1,3}O_{2,3}\right\} .
\end{gather*}
Therefore the change of basis matrix has a structure
\begin{gather*}
\frac{\sin^{-1}(d)}{d\cdot\theta}\sin\theta\cdot A_{SO(N)}=\frac{\sin^{-1}(d)}{d\cdot\theta}\sin\theta\left(\begin{array}{cc}
A_{SO(N-1)} & N\\
0 & P_{SO(N)}
\end{array}\right),
\end{gather*}
where 
\begin{gather*}
P_{SO(N)}=\left(\begin{array}{ccccc}
-\frac{1}{2}\sin\theta & \cos^{2}\frac{\theta}{2} & \cos^{2}\frac{\theta}{2} &  & \cos^{2}\frac{\theta}{2}\\
\cos^{2}\frac{\theta}{2} & \frac{1}{2}\sin\theta & 0 &  & 0\\
 &  & \frac{1}{2}\sin\theta &  & \vdots\\
 &  &  & \ddots & 0\\
 &  &  &  & \frac{1}{2}\sin\theta
\end{array}\right).
\end{gather*}
Therefore $\det P_{SO(N)}=\left(-\frac{1}{2}\right)^{N}\sin^{N}\theta-\cos^{2}\frac{\theta}{2}\det P_{SO(N)}^{\prime}$
where
\begin{gather*}
P_{SO(N)}^{\prime}=\left(\begin{array}{ccccc}
\cos^{2}\frac{\theta}{2} & \cos^{2}\frac{\theta}{2} & \cos^{2}\frac{\theta}{2} &  & \cos^{2}\frac{\theta}{2}\\
0 & \frac{1}{2}\sin\theta & 0 &  & 0\\
0 & 0 & \frac{1}{2}\sin\theta &  & \vdots\\
\vdots & \vdots &  & \ddots & 0\\
0 & 0 &  &  & \frac{1}{2}\sin\theta
\end{array}\right),
\end{gather*}
and is $N-1\times N-1$ matrix. Hence 
\begin{gather*}
\det P_{SO(N)}=-\left(\frac{1}{2}\right)^{N}\sin^{N}\theta-\left(\frac{1}{2}\right)^{N-2}\sin^{N-2}\theta\cos^{4}\frac{\theta}{2}=\\
=-\left(\frac{1}{2}\right)^{N-2}\sin^{N-2}\theta\left(\frac{1}{4}\sin^{2}\theta+\cos^{4}\frac{\theta}{2}\right),
\end{gather*}
and $\det P_{SO(N)}=0$ iff $\theta=0$ or $\theta=\pi$ which finishes
the induction step of the proof. 
\begin{thm}
\label{2-modeuni}Assume $\theta\notin\{0,\,\frac{\pi}{2},\,\pi,\,\frac{3\pi}{2}\}$.
Let 
\begin{gather*}
S=\left\{ O_{k,l}:\, O_{k,l}\in SO(N),\, k,l\in\left\{ 1,\ldots,N\right\} ,\, k<l\right\} ,
\end{gather*}
 where $O_{k,l}$ represent the matrix of a beamsplitter acting on
modes $k$,$l$:
\begin{gather*}
O_{kl}=\cos\theta\left(\kb kk+\kb ll\right)+\sin\theta\left(\kb kl-\kb lk\right).
\end{gather*}
The set S generates $SO(N)$.
\end{thm}

\begin{rem}
In \cite{VY99} it was shown that one can not make arbitrary unitary transformations using only beam splitters when acting on two qubits. This result is not in contradiction with the fact that a $2$-mode beamsplitter generates $SO(4)$. The point is that in the setting of \cite{VY99} `modes' are divided into two pairs and the beamsplitters can only act separately on these pairs. In our setting they are allowed to act on any pair of modes and therefore they can generate more than $SO(2)\times SO(2)$ which happens to be entire $SO(4)$.
\end{rem}

\section{$m$-mode beamsplitters }\label{sec-mmode}

In case of a $2$-mode beamsplitter the freedom stemming from mode
permutations was not significant since it was changing $O_{2}$ into
$O_{2}^{-1}$. For $m$-mode beamsplitters with $m\geq3$ we have
two situations which should be treated differently.

Let $\mathrm{Sym}(m)$ be a permutation group of $m$ elements, $|\mathrm{Sym}(m)|=m!$.
We consider an $m$-mode beamsplitter $ $$O_{m}=e^{A_{m}}\in SO(m)$,
$m\geq3$. Taking into account all possible mode permutations the
starting point is not a single beamsplitter but rather a set of beamsplitters
given by matrices:
\begin{gather*}
S(O_{m}):=\{P_{\sigma}^{T}O_{m}P_{\sigma}:\,\sigma\in\mathrm{Sym}(m)\},
\end{gather*}
where $P_{\sigma}$ is an $m\times m$ permutation matrix corresponding
to $\sigma\in\mathrm{Sym}(m)$. Note that by Schur's lemma $S(O_{m})$
has always at least two elements as the only permutation invariant
matrix in $SO(m)$ is the identity matrix. However, we still can have
two cases:
\begin{enumerate}
\item Trivial action of $\mathrm{Sym}(m)$: $S(O_{m})$ consists of $O_{m}$
and $O_{m}^{-1}$. In this case the action of permutation group is
exactly as for a $2$-mode beamsplitter.
\item Nontrivial action of $\mathrm{Sym}(m)$: $S(O_{m})$ has at least
two non-commuting elements. Combining this with a standard result
of Kuranishi \cite{kuranishi49} saying that the set of pairs that
generate a semisimple Lie group $G$ is open and dense in $G\times G$
then is a good chance set $S(O_{m})$ generates $SO(m)$.
\end{enumerate}
From now on we will say that $O_{m}$ is universal on $k\geq m$ modes
iff $ $the set $ $$S(O_{m})$ is universal on $k\geq m$ modes.
\begin{rem}
If a beamsplitter $O_{m}$ falls into the first case then it cannot
be universal on $m$-modes if only $m$ of them are at our disposal.
On the other hand, if $O_{m}$ belongs to the second case it can be
universal on $m$-modes without usage of any additional modes.
\end{rem}
Having $X\in\mathfrak{so}(m)$ we can embed it into $\mathfrak{so}(m+k)$
in ${m+k \choose m}$ natural ways and similarly having $O\in SO(m)$
we can embed it into $SO(m+k)$ in $ $${m+k \choose m}$ natural
ways by choosing $m$ out of $m+k$ modes and letting $X$ or $O$
operate on them. For example in case when $m=3$ and $k=1$ we have

\begin{scriptsize}

\begin{gather*}
\mathfrak{so}(3)\hookrightarrow\mathfrak{so}(4):\left(\begin{array}{cccc}
\ast & \ast & \ast & 0\\
\ast & \ast & \ast & 0\\
\ast & \ast & \ast & 0\\
0 & 0 & 0 & 0
\end{array}\right),\left(\begin{array}{cccc}
\ast & \ast & 0 & \ast\\
\ast & \ast & 0 & \ast\\
0 & 0 & 0 & 0\\
\ast & \ast & 0 & \ast
\end{array}\right),\left(\begin{array}{cccc}
\ast & 0 & \ast & \ast\\
0 & 0 & 0 & 0\\
\ast & 0 & \ast & \ast\\
\ast & 0 & \ast & 0
\end{array}\right),\left(\begin{array}{cccc}
0 & 0 & 0 & 0\\
0 & \ast & \ast & \ast\\
0 & \ast & \ast & \ast\\
0 & \ast & \ast & \ast
\end{array}\right),\\
SO(3)\hookrightarrow SO(4):\left(\begin{array}{cccc}
\ast & \ast & \ast & 0\\
\ast & \ast & \ast & 0\\
\ast & \ast & \ast & 0\\
0 & 0 & 0 & 1
\end{array}\right),\left(\begin{array}{cccc}
\ast & \ast & 0 & \ast\\
\ast & \ast & 0 & \ast\\
0 & 0 & 1 & 0\\
\ast & \ast & 0 & \ast
\end{array}\right),\left(\begin{array}{cccc}
\ast & 0 & \ast & \ast\\
0 & 1 & 0 & 0\\
\ast & 0 & \ast & \ast\\
\ast & 0 & \ast & 0
\end{array}\right),\left(\begin{array}{cccc}
1 & 0 & 0 & 0\\
0 & \ast & \ast & \ast\\
0 & \ast & \ast & \ast\\
0 & \ast & \ast & \ast
\end{array}\right).
\end{gather*}

\end{scriptsize}

We start with the following simple fact:
\begin{fact}
\label{algebra-embeding}Let $X=\{X_{1},\ldots,X_{k}\}\subset\mathfrak{so}(m)$.
Consider ${m+k \choose k}$ natural embedding of $ $$X$ $ $ into
$\mathfrak{so}(m+k)$. If $ $$X$ generates $\mathfrak{so}(m)$ then
$ $the ${m+k \choose k}$ natural embedding of $ $$X$ into $\mathfrak{so}(m+k)$
generate $\mathfrak{so}(m+k)$.\end{fact}
\begin{proof}
As $X$ generates $\mathfrak{so}(m)$ we have a basis of $\mathfrak{so}(m)$
at our disposal. In particular we can chose it to be a standard basis
$\left\{ E_{i,j}\right\} _{i<j}$, where $1\leq i,\, j\leq m$. By
definition of standard basis, it is obvious that ${m+k \choose k}$
natural embeddings of this basis into $\mathfrak{so}(m+k)$ gives
a set which is an overcomplete basis of $\mathfrak{so}(m+k)$. Therefore
the result follows. \end{proof}
\begin{fact}
\label{spectra}The group $<e^{A_{m}}>$ with $A_{m}\in\mathfrak{so}(m)$
is dense in $\{e^{tA_{m}}:\, t\in\mathbb{R}\}$ iff spectrum of $e^{A_{m}}$
is given by $\{e^{i\phi_{k}}:k\in\{1,\ldots,m\}\}$ with all $\phi_{i}$'s
being irrational multiples of $\pi$.\end{fact}
\begin{proof}
A direct generalization of $2$-mode case. \end{proof}
\begin{thm}
\label{m-n}(Non-trivial action $\mathrm{Sym}(m)$) Assume that for
an $m$-mode beamsplitter $<S(O_{m})>$ is dense in $SO(m)$. Then
$ $$O_{m}$ is universal on $k$ modes for any $k\geq m$. $ $\end{thm}
\begin{proof}
As the set generated by $<S(O_{m})>$ is dense in $ $ $SO(m)$ we
have access to any matrix from $SO(m)$ (at least after closure).
Having arbitrary $SO(m)$ matrix we choose the set $B=\{e^{X_{i}}:\, X_{i}\in\mathfrak{so}(m),\: i\in\{1,...,k\}\}$,
where $X=\{X_{1},\ldots,X_{k}\}$ generate $\mathfrak{so}(m)$ and
spectra of $X_{i}$'s are as in Fact \ref{spectra}. By Fact \ref{algebra-embeding}
$ $the ${m+k \choose k}$ natural embedding of $ $$X$ into $\mathfrak{so}(m+k)$
generate $\mathfrak{so}(m+k)$. As the assumptions of Fact \ref{spectra}
are satisfied we can use Theorem \ref{thm:main} and get the desired
result.\end{proof}
\begin{thm}
(Trivial action of $\mathrm{Sym}(m)$) Assume that for an $m$-mode
beamsplitter $S(O_{m})=\left\{ O_{m},\, O_{m}^{-1}\right\} $ and
the ${m+k \choose k}$ natural embeddings of $S(O_{m})$ into $SO(m+k)$
generate dense set in $SO(m+k)$ for some $k$. Then $ $$O_{m}$
is universal on $l$ modes for any $l\geq m+k$. $ $\end{thm}
\begin{proof}
The analogous to the proof of Theorem \ref{m-n}.\end{proof}
\begin{rem}
(1) All the statements of this section remain true if we substitute
$SO(k)$, $\mathfrak{so}(k)$ with $SU(k)$, $\mathfrak{su}(k)$.
(2) Theorem \ref{m-n} can be used to prove Theorem \ref{2-modeuni}.
We however found the calculation using only the $\mbox{BCH}$ formula
elegant and worth presenting. 
\end{rem}

\section{$3$-mode beamsplitter\label{sec:3-mode-beamsplitter}}

In this section we present the full discussion for $S(O_{3})$. 

The Lie Algebra $\mathfrak{so}(3)$ is generated by $\left\{ E_{12},\, E_{13},\, E_{23}\right\} $.
Let $O_{3}=e^{A_{3}}$ where 
\begin{gather*}
A_{3}=\theta\left(\begin{array}{ccc}
0 & a_{12} & a_{13}\\
-a_{12} & 0 & a_{23}\\
-a_{13} & -a_{23} & 0
\end{array}\right)=\theta\sum_{ij}a_{ij}E_{ij},
\end{gather*}
$a_{ij}\in\mathbb{R}$ and $\sum|a_{ij}|^{2}=1$. The permutation
group consists of six elements 
\begin{gather*}
\mathrm{Sym}(3)=\{e,\,(1,2),\,(1,3),\,(2,3),\,(1,2,3),\,(1,3,2)\}.
\end{gather*}
We want to consider set $S(O_{3})=\{P_{\sigma}^{t}O_{3}P_{\sigma}:\,\sigma\in\mathrm{Sym}(3)\}$.
Making use of $ $$P_{\sigma}^{t}O_{3}P_{\sigma}=e^{P_{\sigma}^{t}A_{3}P_{\sigma}}$,
we have the following $S(A_{3})$ set of Lie algebra elements
\begin{gather*}
A_{3}=a_{12}E_{12}+a_{13}E_{13}+a_{23}E_{23},\, P_{(1,2)}^{T}A_{3}P_{(1,2)}=-a_{12}E_{12}+a_{23}E_{13}+a_{13}E_{23},\\
P_{(1,3)}^{T}A_{3}P_{(1,3)}=-a_{23}E_{12}-a_{13}E_{13}-a_{12}E_{23},\, P_{(2,3)}^{T}A_{3}P_{(2,3)}=a_{13}E_{12}+a_{12}E_{13}-a_{23}E_{23},\\
P_{(1,2,3)}^{T}A_{3}P_{(1,2,3)}=-a_{13}E_{12}-a_{23}E_{13}+a_{12}E_{23},\, P_{(1,3,2)}^{T}A_{3}P_{(1,3,2)}=a_{23}E_{12}-a_{12}E_{13}-a_{13}E_{23}.
\end{gather*}
Note that if $ $$S(O_{3})=\{O_{3},\, O_{3}^{-1}\}$ then $P_{\sigma}^{t}A_{3}P_{\sigma}=\pm A_{3}$
for any $\sigma\in\mathrm{Sym}(3)$. It happens when $a_{12}=-a_{13}=a_{23}$,
i.e. when
\begin{gather*}
A_{3}=\frac{\theta}{\sqrt{3}}\left(E_{12}-E_{13}+E_{23}\right).
\end{gather*}
In the following we divide our discussion of $3$-mode beamsplitters
into two cases according to the behavior under mode permutations.

\subsection{Nontrivial action of $\mathrm{Sym}(3)$}

In this section we consider those beamsplitters whose $A_{3}\neq\frac{\theta}{\sqrt{3}}\left(E_{12}-E_{13}+E_{23}\right)$.
Our goal is to first show that in such a case $S(A_{3})$ generate
the Lie algebra $\mathfrak{so}(3)$. Let us start with the following
simple lemma.
\begin{lem}
\label{le:two -generate}Let $X,Y\in\mathfrak{so}(3)$ be linearly
independent. Then $X$, $Y$ generate $\mathfrak{so}(3)$. \end{lem}
\begin{proof}
Let $\left\{ X_{12},\, X_{13},\, X_{23}\right\} $ be a basis of $\mathfrak{so}(3)$,
where $X_{ij}=\kb ij-\kb ji$. We have:
\begin{gather*}
X=\alpha_{1}X_{12}+\beta_{1}X_{13}+\gamma_{1}X_{23},\\
Y=\alpha_{2}X_{12}+\beta_{2}X_{13}+\gamma_{2}X_{23},\\
\left[X,Y\right]=\left(\gamma_{1}\beta_{2}-\beta_{1}\gamma_{2}\right)X_{12}+\left(\alpha_{1}\gamma_{2}-\gamma_{1}\alpha_{2}\right)X_{13}+\left(\beta_{1}\alpha_{2}-\alpha_{1}\beta_{2}\right)=\\
=-M_{31}X_{12}+M_{32}X_{13}-M_{33}X_{23}
\end{gather*}
We need to show that $X$, $Y$ and $\left[X,Y\right]$ form a basis
of $\mathfrak{so}(3)$. To this end we calculate
\begin{gather}
\det\left(\begin{array}{ccc}
\alpha_{1} & \beta_{1} & \gamma_{1}\\
\alpha_{2} & \beta_{2} & \gamma_{2}\\
-M_{31} & M_{32} & M_{33}
\end{array}\right)=-\left(M_{31}^{2}+M_{32}^{2}+M_{33}^{2}\right)\neq0,\label{eq:10}
\end{gather}
since $X$ and $Y$ are linearly independent (at least one of $M_{ij}\neq0$
).
\end{proof}
Next we proceed with the proof that $S(A_{3})$ generates $\mathfrak{so}(3)$.
There are a few cases to consider.
\begin{enumerate}
\item Exactly one of $a_{ij}$ is non-zero. Without loss of generality assume
that $a_{12}\ne0$ and $a_{13}=0=a_{23}$. In this case elements 
\begin{gather}
A_{3}=a_{12}E_{12},\, P_{(1,3)}^{T}A_{3}P_{(1,3)}=-a_{12}E_{23},\, P_{(2,3)}^{T}A_{3}P_{(2,3)}=a_{12}E_{13},\label{eq:case1}
\end{gather}
form the basis of $\mathfrak{so}(3)$ and therefore matrices $\{P_{\sigma}^{T}A_{3}P_{\sigma}:\sigma\in S_{3}\}$
generate $\mathfrak{so}(3)$. 
\item There are exactly two non-zero $a_{ij}$'s. Assume for example that
$a_{12}=0$. Then 
\begin{gather*}
A_{3}=a_{13}E_{13}+a_{23}E_{23},\, P_{(1,3)}^{T}A_{3}P_{(1,3)}=-a_{23}E_{12}-a_{13}E_{13},
\end{gather*}
 are clearly linearly independent and by Lemma \ref{le:two -generate}
they generate $\mathfrak{so}(3)$.
\item All $a_{ij}$'s are non-vanishing. We need to consider two cases:
(a) $a_{12}\neq-a_{13}\neq a_{23}$ or $a_{12}=-a_{13}\neq a_{23}$
or $a_{12}=a_{23}\neq-a_{13}$ then $A_{3}$ and $ $$P_{(1,2)}^{T}A_{3}P_{(1,2)}$
are linearly independent and by Lemma \ref{le:two -generate} they
generate $\mathfrak{so}(3)$, (b) $a_{12}\neq-a_{13}=a_{23}$ then
$A_{3}$ and $P_{(2,3)}^{T}A_{3}P_{(2,3)}$ are linearly independent
and by Lemma \ref{le:two -generate} they generate $\mathfrak{so}(3)$.
\end{enumerate}
Therefore we have
\begin{fact}
\label{so3-gen}Let $O_{3}=e^{A_{3}}$ be a beamsplitter which admits
nontrivial action of $\mathrm{Sym}(3)$. Then $S(A_{3})$ generates
$\mathfrak{so}(3)$. 
\end{fact}

\subsubsection{The case when $\theta$ is an irrational multiple of $\pi$}

Combining Fact \ref{so3-gen} with theorems \ref{thm:main} and \ref{m-n}
we obtain:
\begin{thm}
Let $O_{3}=e^{A_{3}}$ be a $3$-mode beamsplitter which admits nontrivial
action of $\mathrm{Sym}(3)$ and whose spectrum is given by $\{e^{i\theta},\, e^{-i\theta},1\}$
where $\theta$ is not a rational multiple of $\pi$. Then $S(O_{3})$
generates $SO(3)$ and $O_{3}$ is universal on $k\geq3$ modes.\end{thm}
\begin{rem}
Similar reasoning, however with more cases to consider, can be carried
out for $m$-mode beamsplitters with $m>3$. We will discuss generation
of $\mathfrak{so}(m)$ by $S(A_{m})$ for $m>3$ in a subsequent publication
including analogous calculations for $\mathfrak{su}(m)$.
\end{rem}

\subsubsection{The case when $\theta$ is a rational multiple of $\pi$}

We start with a review of some important and relatively new facts
concerning compositions of rotations in $\mathbb{R}^{3}$. 

Let $O_{1}$ and $O_{2}$ be two finite order rotations about axes
separated by an angle $\alpha$. In the series of papers \cite{CRS00,CRS99,RS98,RS99}
Conway Radin and Sadun studied the group generated by $O_{1}$ and
$O_{2}$ for large class of $\alpha$'s. This group is characterized
by relations that involve generators $O_{1}$ and $O_{2}$. In order
to discuss these relations we recall some basic definitions from algebraic
number theory.
\begin{defn}
A complex number $z\in\mathbb{C}$ is algebraic iff it is a root of
some nonzero polynomial with rational coefficients. If $z$ is not
algebraic then it is called transcendental. 
\end{defn}
Note that the set of algebraic numbers is countable. This can be easily
inferred from the fact that there are countably many coefficients
of polynomials in $\mathbb{Q}[x]$ and each polynomial contributes
finitely many algebraic numbers. Therefore the set of transcendental
numbers is uncountable. 
\begin{thm}
\cite{CRS00} Assume that for $\alpha\in[0,2\pi[$ there are nontrivial
relations between $O_{1}$ and $O_{2}$. Then $e^{i2\alpha}$ is algebraic. 
\end{thm}
As an immediate consequence we get:
\begin{cor}
Assume $e^{i2\alpha}$ is transcendental. Then there are no nontrivial
relations between $O_{1}$ and $O_{2}$ and in particular the group
generated by $O_{1}$ and $O_{2}$ is infinite and dense in $SO(3)$
\end{cor}
As $e^{i2\alpha}$ is generically transcendental the above corollary
covers almost all cases. What is left is the countable set of $ $these
$\alpha\in[0,\,2\pi[$ for which $e^{i2\alpha}$ is algebraic. Note
that if $\alpha$ is a rational multiple of $\pi$ then $e^{2i\alpha}$
is algebraic as it is a root of unity (it satisfies $x^{q}-1=0$ for
some $q\in\mathbb{N}$).
\begin{thm}
\label{thm:25}\cite{RS99} Let $O_{1}$ and $O_{2}$ be two finite
order rotations about axes separated by an angle $\alpha$ which is
a rational multiple of $\pi$. A group generated by $O_{1}$ and $O_{2}$
is infinite and dense in $SO(3)$ with the following three exceptions:
(a) Either $O_{1}=I$ or $O_{2}=I$, (b) $O_{1}^{2}=I$ or $O_{2}^{2}=I$
and $\alpha=\pi/2$, (c)$ $$O_{1}^{4}=I=O_{2}^{4}$ . 
\end{thm}
The exceptions to Theorem \ref{thm:25} correspond to $O_{i}$ being
rotations by: $0$, $\pi/2$ , $\pi$ or $3\pi/2$ that we exclude.
It is well known that when $\alpha$ is a rational multiple of $\pi$
the algebraic order of $e^{i2\alpha}$ can be arbitrary large. It
is therefore interesting to understand better those angles for which
$e^{i2\alpha}$ has, for example, order two. Authors of \cite{CRS00}
do this for the so-called geodetic angles, that is angles whose squared
trigonometric functions are rational. For this kind of angle they
prove relations between $O_{1}$ and $O_{2}$ can occur only for finite
number of $\alpha$'s. The full understanding of all angles for which
$e^{i2\alpha}$ is algebraic is however still non-complete.

Let us return to our problem. We assume $O_{3}=e^{A_{3}}$ is nontrivial
with respect to mode permutations. Then the set $S(O_{3})$ contains
at least two rotations about axes separated by an angle $\alpha$
which is determined by coefficients of $A_{3}$. Making use of the
facts discussed above we have the following:
\begin{lem}
Assume $\alpha$ is a rational multiple of $\pi$ or is such that
$e^{i2\alpha}$ is transcendental. Then $S(O_{3})$ generates $SO(3)$. 
\end{lem}
Combining this with Fact \ref{so3-gen} and Theorems \ref{thm:main}
and \ref{m-n} we get:
\begin{thm}
Let $O_{3}=e^{A_{3}}$ be a beamsplitter which admits nontrivial action
of $\mathrm{Sym}(3)$ and whose spectrum is given by $\{e^{i\theta},\, e^{-i\theta},1\}$
where $\theta$ is a rational multiple of $\pi$. Let $\alpha$ be
the angle between rotation axes of two different elements from $S(O_{3})$.
Assume $\alpha$ is a rational multiple of $\pi$ or is such that
$e^{i2\alpha}$ is transcendental. Then $O_{3}$ is universal on $k\geq3$
modes.
\end{thm}

\subsection{Trivial action of $\mathrm{Sym}(3)$}

In this section we show that when $S(O_{3})=\left\{ O_{3},\, O_{3}^{-1}\right\} $
that is for $O_{3}=e^{A_{3}}$ with
\begin{gather*}
A_{3}=\frac{\theta}{\sqrt{3}}\left(E_{12}-E_{13}+E_{23}\right),
\end{gather*}
the group generated by the four natural embedding of $O_{3}$ into
$SO(4)$ is exactly $SO(3)$. Therefore the beamsplitter given by
$O_{3}$ is not universal on $3$ or $4$ modes. The four embedding
of $O_{3}$ into $SO(4)$ are given by $O_{ijk}=e^{\frac{\theta}{\sqrt{3}}A_{ijk}}$
where
\begin{gather*}
A_{123}=E_{12}-E_{13}+E_{23},\; A_{234}=E_{23}-E_{24}+E_{34},\\
A_{134}=E_{13}-E_{14}+E_{34},\; A_{124}=E_{12}-E_{14}+E_{24}.
\end{gather*}
Elements $A_{ijk}$ are not linearly independent ($A_{123}+A_{134}=A_{124}+A_{234}$)$ $
and one can easily verify that they span the $3$-dimensional subspace
\begin{gather*}
\mathrm{Span}_{\mathbb{R}}\left\{ A_{123},\, A_{234},\, A_{134},\, A_{124}\right\} \subset\mathfrak{so}(4).
\end{gather*}

\begin{lem}
The space $ $$\mathrm{Span}_{\mathbb{R}}\left\{ A_{123},\, A_{234},\, A_{134},\, A_{124}\right\} $
is a $3$-dimensional Lie subalgebra of $\mathfrak{so}(4)$.\end{lem}
\begin{proof}
It is enough to show that $\mathrm{Span}_{\mathbb{R}}\left\{ A_{123},\, A_{234},\, A_{134},\, A_{124}\right\} $
is closed under Lie bracket. The result follows from the commutations
relations
\begin{gather*}
\left[A_{123},\, A_{234}\right]=A_{134}+A_{124},\;\left[A_{123},\, A_{134}\right]=A_{124}-A_{234},\\
\left[A_{123},\, A_{124}\right]=-A_{234}-A_{134},\;\left[A_{234},\, A_{134}\right]=A_{123}+A_{124},\\
\left[A_{234},\, A_{124}\right]=A_{123}-A_{134},\;\left[A_{134},\, A_{124}\right]=A_{123}+A_{234}.
\end{gather*}

\end{proof}
Let us remind that $\mathfrak{so}(4)\simeq\mathfrak{so}(3)\oplus\mathfrak{so}(3)$.
It is therefore natural to suspect that \begin{gather}\mathrm{Span}_{\mathbb{R}}\left\{ A_{123},\, A_{234},\, A_{134},\, A_{124}\right\} \simeq\mathfrak{so}(3).\end{gather} 
\begin{lem}
\label{lem:so3}The space $ $$\mathrm{Span}_{\mathbb{R}}\left\{ A_{123},\, A_{234},\, A_{134},\, A_{124}\right\} $
is isomorphic to $\mathfrak{so}(3)$. \end{lem}
\begin{proof}
Let
\begin{gather*}
X=\frac{1}{4}\left(A_{123}+A_{234}+A_{134}+A_{124}\right),\\
Y=\frac{1}{4}\left(A_{123}+A_{234}-A_{134}-A_{124}\right),\\
Z=-\frac{1}{4}\left(A_{123}-A_{234}-A_{134}+A_{124}\right).
\end{gather*}
It is easy to verify that $\mathrm{Span}_{\mathbb{R}}\left\{ X,\, Y,\, Z\right\} =\mathrm{Span}_{\mathbb{R}}\left\{ A_{123},\, A_{234},\, A_{134},\, A_{124}\right\} $.
On the other hand we have:
\begin{gather*}
\left[X,\, Y\right]=Z,\ \left[Z,\, X\right]=Y,\ [Y,Z]=X,
\end{gather*}
which are commutation relations of $\mathfrak{so}(3)$. 
\end{proof}
Combining Lemma \ref{lem:so3} with Theorem \ref{thm:main} we get
\begin{thm}
\label{thm:irr}Let $\theta$ be an irrational multiple of $\pi$.
The group generated by four natural embeddings of $O_{3}=e^{A_{3}}$
into $SO(4)$ where
\begin{gather*}
A_{3}=\frac{\theta}{\sqrt{3}}\left(E_{12}-E_{13}+E_{23}\right),
\end{gather*}
is isomorphic to $SO(3)$ and therefore $O_{3}$ is not universal
on $4$ modes.
\end{thm}
We are left with the case of $\theta$ which is a rational multiple
of $\pi$. By Lemma \ref{lem:so3} the rotation matrices $O_{ijk}$
act on some $3$-dimensional subspace of $\mathbb{R}^{4}=\mathrm{Span}_{\mathbb{R}}\left\{ e_{1},\, e_{2},\, e_{3},\, e_{4}\right\} $.
Note that, for example, $O_{123}$ and $O_{234}$ are rotations by
$\theta$ about axes given by:
\begin{gather*}
\vec{n}_{123}=-\frac{1}{\sqrt{3}}e_{1}+\frac{1}{\sqrt{3}}e_{2}-\frac{1}{\sqrt{3}}e_{3},\\
\vec{n}_{234}=-\frac{1}{\sqrt{3}}e_{2}+\frac{1}{\sqrt{3}}e_{3}-\frac{1}{\sqrt{3}}e_{4}.
\end{gather*}
Let $\alpha$ be the angle between $\vec{n}_{123}$ $ $and $\vec{n}_{234}$.
One has $ $$\cos(\alpha)=-\frac{2}{3}$ and therefore $\alpha$ is
an geodetic angle - an angle whose squared trigonometric functions
are rational. The Primordial Theorem (Theorem 2 of \cite{CRS00})
lists all geodetic angles $\alpha=\sin^{-1}\left(\sqrt{\frac{p}{q}}\right)$
with $p$ and $q$ coprime that support nontrivial relations between
two finite order rotations about axes separated by $\alpha$. In our
case $\sin(\alpha)=\sqrt{\frac{5}{9}}$ and it does not belong to
the list given in \cite{CRS00}. Therefore
\begin{thm}
\label{thm:rat}Let $\theta$ be a rational multiple of $\pi$. The
group generated by four natural embeddings of $O_{3}=e^{A_{3}}$ into
$SO(4)$ where
\begin{gather*}
A_{3}=\frac{\theta}{\sqrt{3}}\left(E_{12}-E_{13}+E_{23}\right),
\end{gather*}
is isomorphic to $SO(3)$ and therefore $O_{3}$ is not universal
on $4$-modes.\end{thm}
\begin{rem}
One can also show that $O_{3}$ is not universal on $5$-modes - the
natural ten embeddings of $O_{3}$ into $SO(5)$ generate the group
isomorphic to $SO(4)$. We conjecture that on $k$-modes where $k\geq4$
the ${k \choose 3}$ natural embeddings of $O_{3}$ generate $SO(k-1)$. 
\end{rem}
\begin{gather*}
\end{gather*}

\section{$ $Summary and outlook}

In this paper we discussed the universality problem of $m$-mode real
beamsplitters for $m=2,3$ from the perspective of control theory
using some nice properties of the $SO(3)$ group. We also pointed
out the importance of the set $S(O_{m})$ which is the orbit of adjoint
action of permutation group through $O_{m}$. In particular we showed
that when $S(O_{3})=\left\{ O_{3},\, O_{3}^{-1}\right\} $ the beamsplitter
$O_{3}$ is not universal on both $3$ and $4$ modes (we also know
it is the case for $5$-modes). The study of similar phenomena in
higher dimensions is a natural direction we want to explore. The other
problem would be extension of the result presented here to complex
beamsplitters. This requires proving several nontrivial results that
we plan to discuss elsewhere.

\section*{Acknowledgments}

I would like to thank Jan Gutt for the inspiring email correspondence
concerning finite generation of Lie groups and Lorenzo Sadun for referring
me to papers \cite{CRS00,CRS99,RS99,RS98}. Moreover, I thank Scott
Aaronson, Adam Bouland, Aram Harrow, Jon Keating, Marek Ku\'s, John
Mackay, Jonathan Robbins, Cyril Stark for discussions, Nick Jones
for reading the manuscript and the anonymous referee for suggestions that led to improving the contents of this paper. The author is supported by the Marie Curie
International Outgoing Fellowship.

\nonumsection{References}

\end{document}